\documentclass[a4paper,UKenglish,cleveref, autoref]{lipics-v2019} %
\usepackage{ stmaryrd }
\usepackage{bbm}
\usepackage{environ}

\NewEnviron{display}{\\[1mm] \phantom{.{}\hspace{0.3cm}{}.}$\BODY$\vspace{1mm}}

\usepackage{todonotes}
\usepackage{macros}

\usepackage{mathpartir}
\usepackage{amsmath,amssymb}

\hideLIPIcs
\nolinenumbers

\title{A certifying extraction with time bounds from Coq to call-by-value $\lambda$-calculus}

\author{Yannick Forster}{Saarland University, Saarland Informatics Campus (SIC), Saarbrücken, Germany}{forster@ps.uni-saarland.de}{}{}%

\author{Fabian Kunze}{Saarland University, Saarland Informatics Campus (SIC), Saarbrücken, Germany}{kunze@ps.uni-saarland.de}{}{}%

\authorrunning{Y. Forster and F. Kunze}

\Copyright{Yannick Forster and Fabian Kunze}%

\usepackage[countsame]{coqtheorem}
\setBaseUrl{{https://uds-psl.github.io/certifying-extraction-with-time-bounds/website/L.}}
\usepackage{xcolor}
\usepackage{listings}
\newcommand{\term}{\ensuremath{\mathbf{T}}\xspace}
\usepackage{lstcoq}
\let\coq\lstinline

\newcommand{\coqm}[1]{\text{\coq{#1}}}

\usepackage{xspace}



\supplement{The Coq development is accessible at \\ \url{https://github.com/uds-psl/certifying-extraction-with-time-bounds}}%

\EventEditors{John Harrison, John O'Leary, and Andrew Tolmach}
\EventNoEds{3}
\EventLongTitle{10th International Conference on Interactive Theorem Proving (ITP 2019)}
\EventShortTitle{ITP 2019}
\EventAcronym{ITP}
\EventYear{2019}
\EventDate{September 9--12, 2019}
\EventLocation{Portland, OR, USA}
\EventLogo{}
\SeriesVolume{141}
\ArticleNo{10}

\renewcommand\L{\ensuremath{\mathsf{L}}\xspace}

\begin{document}

\maketitle

\begin{abstract}
  We provide a plugin extracting Coq functions of simple polymorphic types to the (untyped) call-by-value $\lambda$-calculus~\L.
  The plugin is implemented in the MetaCoq framework and entirely written in Coq.
  We provide Ltac tactics to automatically verify the extracted terms w.r.t a logical relation connecting Coq functions with correct extractions and time bounds, essentially performing a certifying translation and running time validation.
  We provide three case studies:
  A universal \L-term obtained as extraction from the Coq definition of a step-indexed self-interpreter for \L, a many-reduction from solvability of Diophantine equations to the halting problem of \L, and a polynomial-time simulation of Turing machines in \L.
\end{abstract}

\section{Introduction}

Every function definable in constructive type theory is computable in a model of computation.
This also enables many proof assistants based on constructive type theory to implement extraction into a ``real'' programming language.
On the more foundational side, various realisability models for fragments of constructive type theory increase the trust in this meta-theorem, because realisers for types are the codes of computable functions.

The computability of all definable functions also enables the study of synthetic computability theory in constructive type theory~\cite{forster2019synthetic, bauer2006first}.
For instance, one can define decidability by \coq|dec P := exists f, forall x, P x <-> f x = true| and no reference to a concrete model of computation is needed.
The undecidability of a predicate $p$ can be shown by defining a many-one reduction from the halting problem of Turing machines to $p$ in Coq, again without referring to a concrete model.
The computability of all definable functions can, however, not be proved inside the type theory itself, similar to other true statements like parametricity.
At the same time, for every concrete defined function of the type theory, one can always prove computability as theorem in the type theory.
Given for instance any concrete function $f : \nat \to \nat$ definable in constructive type theory, one can construct a term of the $\lambda$-calculus $t_f$ s.t. for all $n : \nat$, there is a proof in the type theory that $t_f ~ \overline n$ reduces to $\overline{f n}$ (where $\overline{\vphantom t \cdot}$ is a suitable encoding of natural numbers).
The construction of $t_f$ from $f$ is relatively simple, since it is syntax-directed and the terms of type theory are just (possibly type-decorated) terms of an expressive untyped $\lambda$-calculus.
Another way to see this construction is as extraction from the type theory into the $\lambda$-calculus.

We implement one such construction of $\lambda$-terms $t_f$ for a certain subset of type theory:
We use the MetaCoq framework~\cite{anand2018towards} to extend the proof assistant Coq with a command to extract Coq functions of simple polymorphic types into the weak call-by-value $\lambda$-calculus \L and provide tactics to automatically prove the correctness of the term.
In addition to the correctness, our extraction command can generate recurrence equations that, if instantiated with a function by the user, describe the time complexity as number of $\beta$-steps of the extracted $\lambda$-term on its arguments.
Our target calculus \L has been used before to formalise computability theory in Coq~\cite{FS}.
Since it is (syntactically) the pure $\lambda$-calculus, recursive functions have to be encoded using a fixed-point combinator and inductive types using Scott's encoding.

Our extraction has several use cases:

First, while parts of computability theory can be formalised in Coq without referring to a model of computation~\cite{forster2019synthetic}, one needs a deep embedding of computable functions to e.g. construct universal machines.
Our framework then allows the user to write all functions in Coq and automatically get $\lambda$-terms computing them, similar to practice on paper where function in the model are never spelled out.
For instance, the automated construction of a universal $\lambda$-term takes about 30 lines and no manual proofs, whereas by hand construction and verification take about 500 lines~\cite{FS}.

Second, to the best of our knowledge, there are no formalisations of computational complexity theory in any proof assistant.
We hope that our framework can be used to enable formalisations of basic complexity theory.
One tedium -- even on paper -- when doing complexity theory in a way such that all details are spelled out is that constructing and verifying functions in the chosen model of computation is hard.
With our framework, this burden is significantly lowered: Implementations can be given in Coq and only a suitable running-time function has to be given by hand.
We extract a definition of Turing machines to show that \L can simulate $k$ steps of a Turing machine in a number of $\beta$-steps linear in $k$.

Third, synthetic undecidability and the notion of synthetic decidability and enumerability have been analysed in Coq~\cite{FHS, forster2019synthetic, forster2019certified, Larchey-WendlingForster:2019:H10_in_Coq}.
This resulted in a library of undecidable problems in Coq~\cite{forster2018towards}.
All problems of the library are shown undecidable by reduction from the halting problem of Turing machines.
To show that all contained problems are actually interreducible with the halting problem, one has to give many-one reductions from the problems to the halting problem.
Using extraction, a reduction to the halting problem for \L is straightforward: It suffices to prove enumerability in Coq, which follows a clear scheme, and then extract the Coq enumerator automatically to \L.
We demonstrate the power of this method by reducing solvability of Diophantine equations to the halting problem of \L.

Lastly, it might be beneficial to use classical axioms like choice when verifying reductions.
Since the computability of all definable functions does not necessarily hold given classical assumptions, one can extract the used reductions to \L to ensure their computability.

\textsf{\textbf{Related Work}}
Myreen and Owens~\cite{myreen2014proof} implement a proof-producing translation from the higher-order logic implemented in the HOL4 system with a state-and-exception monad into CakeML~\cite{kumar2014cakeml}.
The translation also produces proofs for the translated terms, similar to our approach.
Hupel and Nipkow~\cite{hupel2018verified} give a verified compiler from a deep embedding of Isabelle/HOL to CakeML.
Similar to our work, they use a logical relation to connect Isabelle definitions to an intermediate representation.

Mullen et al.~\cite{mullen2018oeuf} provide a verified compiler from a subset of Coq to assembly.
Anand et al.~\cite{anand2017certicoq} report on ongoing work on verifying the full extraction process of Coq, also based on the MetaCoq framework.
They extract Coq functions into Clight, an intermediate language of the CompCert compiler, and are thus able to obtain verified assembly code for Coq functions.
Letouzey~\cite{letouzey2004certified} describes the theoretical foundations of extraction in Coq.
Our logical relation can be seen as a light-weight version of his simulation predicate for simple polymorphic types.

Köpp~\cite{koepp2018} verifies program extraction for functions in the Minlog proof assistant into a $\lambda$-calculus-like system.

Guéneau et al.~\cite{gueneau2018fistful} verify the asymptotic complexity of functional programs in Coq, based on separation logic with time credits.

We have reported on a preliminary version of our extraction plugin in~\cite{CoqWS}.

\section{The call-by-value $\lambda$-calculus \L}
\label{sec:L}

We use the weak call-by-value $\lambda$-calculus \L defined in~\cite{FS} and based on~\cite{Plotkin75, DalLagoMartini08} as target language.
It comes with an inductive type of terms
\begin{align*}
  s,t,u,v
  &~:~\term~::=~
    n\mid st\mid\lambda s
  \qquad(n:\nat)
\end{align*}
and a recursive function \emph{$\subst sku$}
providing a simple, capturing \emph{substitution} operation:
\begin{align*}
  \subst kku
  &~:=~u
  &
  \subst nku
  &~:=~n
  \tag{if $~n\neq k$}
  \\
  \subst{(st)}ku
  &~:=~(\subst sku)(\subst tku)
  &
  \subst{(\lambda s)}ku
  &~:=~\lambda(\subst s{1 + k}u)
\end{align*}

We will freely switch between a named representation for examples and the representation using de Bruijn indices for definitions, i.e.\ we write $\lambda x y.x$ for $\lambda\lambda 1$.

We define an inductive weak call-by-value
\emph{reduction relation $s\red t$}:
\begin{mathpar}
  \inferrule*{~}{(\lambda s) (\lambda t)\red\subst s0{\lambda t}}
  \and
  \inferrule*{s\red s'} {st\red s't}
  \and
  \inferrule* {t\red t'} {s t\red s t'}
\end{mathpar}
We write $\red^*$ for the reflexive transitive closure of $\red$,
$\red^k$ for exactly and $\red^{\leq k}$ for at~most~$k$~steps.

Note that -- contrary to Coq reduction -- \L-reduction does not apply below binders.
Due to the capturing substitution relation, reduction is only well-behaved on closed terms.
We call a term \emph{closed} if it has no free variables. Closed
abstractions are called \emph{procedures} and are the (only) normal forms of normalising, closed terms.

\L provides for recursion using a fixed-point operator:
\setCoqFilename{Tactics.Lsimpl}
\begin{lemma}[Fact 6~\cite{FS}][rho_correct]\label{lem:rho}
  There is a function $\rho : \term \to \term$ s.t. $(\rho u)v$ reduces to $u(\rho u)v$ for procedures $u, v$.
\end{lemma}

Inductive datatypes can be encoded using Scott encodings~\cite{Mogensen1992,Jansen2013}, which we explain in \Cref{sec:extract-constructor}.

One crucial property of \L reduction is that it is uniformly confluent, making every reduction to a normal form have the same length:

\begin{theorem}[Corollary 8 \cite{FS}]
  If $s \red^{k_1} v_1$, $s \red^{k_2} v_2$ for procedures $v_i$, then $v_1 = v_2 \land k_1 = k_2$.
\end{theorem}
\newcommand\unit{\mathbbm{1}}
For the remainder of this paper, we will write $\Type$ for the type of types in Coq, $\Prop$ for the type of propositions, $\List X$ and $\Option X$ for lists and options over $X$, and $\unit$ (with $\star : \unit$) for the unit~type.

\section{Correctness and time bounds}
\label{sec:corr}
\setCoqFilename{Tactics.Computable}
\newcommand{\computes}[3][]{#3 \sim^{#1} #2}
\newcommand{\encf}[1]{\varepsilon_{#1}}

We define when a term computes a Coq function using two logical relations, one considering just correctness, and one correctness with time bounds.
Crucial for both definitions is the notion of an encoding function:

\begin{definition}[][registered]
  A function $\encf A : A \to \term$ is an encoding function for a type  $A$ if $\encf A$ is injective and only returns procedures.
\end{definition}

Notice that the only types where such a function can be defined are computationally relevant (i.e. non-propositional), countable types like $\bool$, $\nat$, $\Option X$,
or $\List X$ over countable $X$.

\subsection{Correctness}

We define a logical relation $\computes a {t_a}$, meaning the \L-term $t_a$ correctly computes $a$.
We will only define this predicate for elements $a : A$ where $A$ is a simple type of the form $A_1 \to \dots \to A_n$. 

We define the predicate $\computes a {t_a}$ as follows:
\begin{mathpar}
  \infer{~}{\computes a {\encf A a}} \quad (\text{for } a : A)
  
  \infer{ t_f \text{ is a procedure} ~\land \\\\ \forall a t_a.~ \computes a {t_a} \to 
    \Sigma v : \term.~ t_f
    t_a \red^* v \land \computes {f a} v    }{\computes f {t_f}} \quad  (\text{for $f : A \to B$})
\end{mathpar}

For elements $a : A $ for encodable types $A$ the only term computing them is their encoding.
Functions $f : A \to B$ are computed by a procedure $t_f$, if for every $a : A$ computed by $t_a$ the term $t_f\ t_a$ computes $f a$.
Note that we could alternatively define the first rule s.t.~every term $t$ convertible to the encoding $\encf A a$ computes $a$, and then simplify the second rule to read $\computes {f a}{t_f t_a}$.
While technically correct, this simplification does not work for the extension of the relation with time complexity.
We thus stick with the more complicated second rule where we require a term $v$ (using the type theoretical sum $\Sigma$\footnote{for non type-theoriest, $\Sigma$ can be read as a computable existential quantifier}) s.t. $t_f t_a$ reduces to $v$ and $v$ computes $f a$.

Defining this predicate in Coq is not entirely straightforward.
As common when defining logical relations, the definition is not strictly positive and thus not accepted by Coq as inductive predicate.
The standard approach for non strictly positive predicates is to translate them into a recursive function.
However, here we would need recursion over types, which is not supported in Coq's type theory.%
\newcommand\TT{\mathfrak{T}}
We circumvent this restrictions by defining a type former $\TT : \Type \to \Type$ capturing exactly the types we want to recurse on and  define the predicate by recursion on $\texttt{ty} : \TT A$:
\begin{lstlisting}
Inductive $\TT$ : Type -> Type :=
  $\TT$_base A `{registered A} : $\TT$ A (* base types *)
| $\TT$_arr A B ($\texttt{ty}_1$ : $\TT$ A) ($\texttt{ty}_2$ : $\TT$ B) : $\TT$ (A -> B). (* functions types *)

Fixpoint computes {A} ($\texttt{ty}$ : $\TT$ A) {struct $\texttt{ty}$}: A -> L.term -> Type :=
  match $\texttt{ty}$ with
    $\TT$_base => fun x ext => (ext = enc x)
  | $\TT$_arr A B $\texttt{ty}_1$ $\texttt{ty}_2$ =>  fun f t_f  => proc t_f * (* t_f is closed and normal *)
                                      forall (a : A) t_a, computes $\texttt{ty}_1$ a t_a ->
                                      {v : term &  (* there exist$$s a term v*)
                                      (t_f t_a >* v)* computes $\texttt{ty}_2$ (f a) v} end.
\end{lstlisting}
The first constructor of $\TT$ takes every encodable type as argument, denoted in Coq by the \coq{registered} type class, which we explain in \Cref{sec:gen-enc}.
The second constructor captures exactly non-dependent functions.
The definition of \coq{computes} then exactly captures the inductive rules given above\rlap.\footnote{Note that $\texttt{\{}$\coq!v : term $\,\,\texttt{\&}\,\,$ P v !$\texttt{\}}$ is Coq-notation for a dependent pair.}
By making $\TT$ a type class, instances $\texttt{ty}$ can always be obtained automatically.

As a running example, we will use the function \coq{map}$~{X~Y} : (X \to Y) \to \List X \to \List Y$ on lists for fixed types $X$ and $Y$.
We assume that $X$, $Y$, $\List X$ and $\List Y$ are all encodable.
Then $\computes {\texttt{map}~{X~Y}} t$ is equivalent to $t$ being a procedure and the proposition $\forall (f : X \to Y) (t_f : \term) (L : \List X).~ \computes f {t_f} \to t\ t_f\, (\encf {} L) \red^* \encf{} (\texttt{map}~{X~Y}f\ L)$.

Note that $\computes {} {}$ is defined similarly to $\llbracket \cdot \rrbracket_{\mathbbm{2}}$ on inductives and functions in~\cite{letouzey2004certified}.

\subsection{Time bounds}\label{sec:time-bounds}
\newcommand{\timec}[1]{\tau_{#1}}

We extend the computability predicate to include time bounds.
As time measure for a term we use its number of $\beta$-steps to a normal form, which is shown reasonable in~\cite{FKR}.
The time bound is expressed depending on the input itself, not its size: e.g.\ for $f : \List \nat \to \bool$ with $\computes f {t_f}$, we want to have a time complexity function $\timec f : \List \nat \to \nat$ such that $\forall L : \List \nat.~ t_f (\encf{} L) \red^{\leq (\timec f L)} \encf{} (f L)$.

We generalise this idea to also account for higher-order functions and define the type $\mathcal{C}$ of complexity measures $\timec a$ for $a : A$ as follows:
\begin{display}
  \mathcal{C} A := \unit  \quad\quad  \mathcal{C} (A \to B) := A \to \mathcal{C} A \to \nat \times \mathcal{C} B
\end{display}

Given the term \coq{map}$~{X~Y}$ of type $(X \to Y) \to \List X \to \List Y$ as above, its complexity measure~$\timec {\coqm{map}~{X~Y}}$ will be
$ (X \to Y) \to (X \to \unit \to \nat \times \unit) \to \nat \times (\List X \to \unit \to \nat \times \unit)$,
which is equivalent to
$(X \to Y) \to (X \to \nat) \to \nat \times (\List X \to \nat)$,
i.e.\ it is a function that, given an argument $f : X \to Y$ and a complexity measure $\timec f : \mathcal{C}(X \to Y)$ (being equivalent to $X \to \nat$), returns a pair of the number of steps $\texttt{map}\, f$ needs to (partially) evaluate, and a function that for $L : \List X$ computes the remaining number of steps $\texttt{map}\, f\, L$ needs to evaluate.

We can extend the computability predicate with time bounds into a predicate $\computes[\timec{a}] a {t_a}$:%
\begin{mathpar}
  \infer{ ~ }{\computes[\timec {}] a {\encf A a}} \quad (\text{for } a : A)
  
  \infer{t_f \text{ is a procedure}~\land \\\\ \forall a t_a \timec a.~ \computes[\timec a] a {t_a} \to \Sigma v : \term. \\\\
    t_f t_a \red^{\leq n} v \land \computes[\timec{}]{f a} v
    \text{ where } \timec f a \timec a = (n, \timec {})    }{\computes[\timec a] f {t_f}} \quad  (\text{for $f : A \to B$})
\end{mathpar}

The first rule is essentially unchanged:
Since encoded terms $\encf A a$ are always normal, $\computes[\timec {}] a {\encf A a}$ holds for every complexity measure $\timec {}$.
For the second rule, we decompose $\timec f a \timec a$ into $n$ and $\timec {}$.
The complexity measure $\timec {} : \mathcal{C} B$ is the complexity measure for $\computes[\timec{}] {f a} v$ and $n$ is the number of steps $t_f~t_a$ needs to reach $v$.

Similar to before, we implement the predicate by recursion on an element of $\TT A$:
\begin{lstlisting}
Fixpoint computesTime {A} (ty : $\TT$ A) {struct ty}: A -> L.term ->$\mathcal{C}$ A -> Type := (* ... *).
\end{lstlisting}

\section{Extraction}
\label{sec:extract}

We describe the different tools needed to extract functions, constructors and to generate encoding functions.

\subsection{Template-Coq}
\label{sec:template-coq}

Template-Coq is a quoting library for Coq, now part of the MetaCoq project and originally developed by Malecha~\cite{malecha2015extensible}.
The current state of the project is explained by Anand et al.\ \cite{anand2018towards} and Boulier~\cite{boulier2018extending}.

Template-Coq provides an inductive type \coq{term} implementing the abstract syntax of Coq as an inductive type (\cref{fig:template-coq-term}).
It comes with a monad \coq{TemplateMonad : Type -> Prop} (\cref{fig:template-coq}) which allows operations like quoting (i.e.\ converting Coq terms into their abstract syntax tree), unquoting (i.e.\ converting abstract syntax trees into Coq terms), evaluating terms, and making definitions.
An operation \coq{m : TemplateMonad A} can be executed using the \coq{Run TemplateProgram m} vernacular command.

As an example, the following function obtains the type of its input by unquoting it into a pair of a type and an element, projecting out the type and returning its quotation:
\begin{lstlisting}
Definition tmTypeOf (s : term) :=
  u <- tmUnquote s ;;
  u'<-  tmEval hnf (my_projT1 u) ;;
  t <- tmQuote u';;
  ret t
\end{lstlisting}

\begin{figure}[h]
  \centering
  \begin{subfigure}{\textwidth}
\begin{lstlisting}
Inductive term : Set :=
| tRel       : nat -> term
| tLambda    : name -> term (* the type *) -> term -> term
| tLetIn     : name -> term (* the term *) -> term (* the type *) -> term -> term
| tApp       : term -> list term -> term
| tConst     : kername -> universe_instance -> term
| tConstruct : inductive -> nat -> universe_instance -> term
| tCase      : (inductive * nat) (* num of parameters *) -> 
                term (* type info *) -> term (* discriminee *) ->
                list (nat * term) (* branches *) -> term
| tFix       : term -> nat -> term
(* ... *).
\end{lstlisting}
    \subcaption{Term representation}
    \label{fig:template-coq-term}
  \end{subfigure}
    \begin{subfigure}{\textwidth}
\begin{lstlisting}
Inductive TemplateMonad : Type -> Prop :=
(* Monadic operations *)
| tmReturn : forall {A:Type}, A -> TemplateMonad A
| tmBind : forall {A B : Type}, TemplateMonad A -> 
               (A -> TemplateMonad B) -> TemplateMonad B
(* General commands *)
| tmPrint : forall {A:Type}, A -> TemplateMonad unit
| tmFail : forall {A:Type}, string -> TemplateMonad A
| tmEval : reductionStrategy -> forall {A:Type}, A -> TemplateMonad A
(* Return the defined constant *)
| tmDefinitionRed : ident -> option reductionStrategy -> forall {A:Type}, A -> TemplateMonad A
| tmLemmaRed : ident -> option reductionStrategy -> forall A, TemplateMonad A
(* Quoting and unquoting commands *)
| tmQuote : forall {A:Type}, A  -> TemplateMonad term
| tmUnquote : term  -> TemplateMonad {T : Type & T}
| tmUnquoteTyped : forall A, term -> TemplateMonad A
\end{lstlisting}
    \subcaption{Monad operations}
    \label{fig:template-coq}
  \end{subfigure}
\caption{Template-Coq's definitions}
\end{figure}

\subsection{Extracting Terms}\label{sec:extract-term}

We define a monadic function \coq{extract} which can extract admissible Coq terms into \L.
In order to extract a Coq term, all the constants appearing in it have to be extracted.
To save work, we remember previously generated extracts, similar to Anand et al.~\cite{anand2018towards}, who use explicit dictionaries for this task.
We employ Coq's type class mechanism instead of dictionaries:
\begin{lstlisting}
Class extracted {A : Type} (a : A) := int_ext : L.term.
\end{lstlisting}
This also defines a function \coq|int_ext| which allows referring to the extracted term corresponding to \coq{a} as \coq{int_ext a}, if it exists, and otherwise get an error.

We restrict the terms we can extract to admissible terms:

\begin{definition}\label{def:adm}
A type $A$ is admissible if $A$ is of the form $\forall X_1 \dots X_n : \Type.~B_1 \to \dots \to B_m$ with $B_m \neq \Type$.
Terms $a : A$ are admissible if $A$ is admissible and if all constants $c : C$ that are proper subterms of $a$ are either
\begin{enumerate}
\item admissable and occur syntactically on the left hand side of an application fully instantiating the type-parameters of $c$ with constants or
\item of type $\Type$ and occur syntactically on the right hand side of an application instantiating type parameters.
\end{enumerate}
\end{definition}

This means a type $A$ is admissible if it has no quantification over terms, quantification over types in $A$ is in prenex normal form and the return type of $A$ is not $\Type$.
The function \coq{map} (\cref{fig:map_vec}) for instance is admissible.
The only constants appearing in its body are \coq{nil} and \coq{cons}, which are both admissible and occur fully instantiated.

\begin{figure}
  \centering
\begin{lstlisting}
Definition map (A B : Type) : (A -> B) -> list A -> list B := fun f =>
  fix map := match l with | [] => @@nil B | a :: t => @@cons B (f a) (map l) end
\end{lstlisting}
  \caption{Definition of \coq{map : forall A B : Type, (A -> B) -> list A -> list B}}
  \label{fig:map_vec}
\end{figure}

We define an extraction function which correctly extracts admissible
terms of a type without type-parameters. If we want to extract
polymorphic functions like \coq{map} we use Coq's section
mechanism and fix the types \coq{A} and \coq{B} as section variables and extract \coq{map A B}. 

The type of the extraction function is
\begin{display}
  \quad\quad\quad\quad\text{\coq{extract : (nat -> nat) -> term -> nat ->  TemplateMonad L.term}}
\end{display}

The first argument is an environment argument which tracks lifting
information for de Bruijn indices for the treatment of fixed points.
The last argument is a fuel argument, needed because recursion on the
right-hand constituents of an application is not structurally
recursive.

Dealing with variables and binders is relatively straightforward,
since Template-Coq already uses a de Bruijn representation of
terms. Variables translate directly to variables, functions to
$\lambda$ and fixed points can be translated using $\rho$ from \ref{lem:rho}.
We have to lift variables when entering an abstraction using the standard de Bruijn lifiting operation ($\uparrow$):
\begin{lstlisting}
Notation "lift  E" := (fun n => match n with 0 => 0 | S n => S (E n) end).

Fixpoint extract env s fuel :=
  match fuel with 0 => tmFail "out of fuel" | S fuel =>
  match s with
    Ast.tRel n => t <- tmEval cbv (var (env n));; ret t
  | Ast.tLambda _ _ s => t <- extract (lift env) s fuel ;; ret (lam t)
  | Ast.tFix [BasicAst.mkdef _ nm ty s _] _ =>  
      t <- extract (fun n => S (env n)) (Ast.tLambda nm ty s) fuel ;;  ret (rho t)
\end{lstlisting}

In order to extract applications \coq{s R} (where \coq{R} is a list of all arguments), we count the number of type parameters of \coq{s}.
If it has none, extraction is straightforward recursion.
We extract \coq{s R} by folding over the list \coq{R} as the application of the extraction of all subterms:
\begin{lstlisting}
  | Ast.tApp s R =>
    p <- tmDependentArgs s;;
    if p =? 0 then
      t <- extract env s fuel;;
      monad_fold_left (fun t1 s2 => t2 <- extract env s2 fuel ;; ret (app t1 t2)) R t
\end{lstlisting}
If \coq{s} has \coq{p > 0} type parameters, we assume that it is the syntax of a previously extracted constant.
We split \coq{R} into type parameters \coq{P} and the list of computational arguments \coq{L} and unquote \coq{tApp s P} as \coq{a}.
We then obtain an extraction \coq{t} for the constant \coq{a} using the \coq{tmTryInfer} operation invoking type class search.
Finally, we again recursively extract by folding over the list of arguments \coq{L}:
\begin{lstlisting}
    else
      let (P, L) := (firstn p R, skipn p R)  in
      s' <- tmEval cbv (Ast.tApp s P);;
      (if closedn 0 s' 
         then ret tt 
         else tmFail "The term contains variables as type parameters.");;
      a <- tmUnquote s' ;;
      a' <- tmEval cbn (my_projT2 a);;
      n <- (tmEval cbv (String.append (name_of s) "_term") >>= tmFreshName) ;;
      i <- tmTryInfer n (Some cbn) (extracted a') ;;
      let t := (@@int_ext _ _ i) in
      monad_fold_left (fun t1 s2 => t2 <- extract env s2 fuel ;; ret (app t1 t2)) L t                             
\end{lstlisting}

For all other syntactic constructs we refer to the Coq code.

We wrap the extraction function into an operation which adds definitions:
\begin{lstlisting}
Definition tmExtract (nm : option string) {A} (a : A) : TemplateMonad L.term :=
  q <- tmUnfoldTerm a ;;
  t <- extract (fun x => x) q FUEL ;;
  match nm with
    Some nm => nm <- tmFreshName nm ;;
                 @@tmDefinitionRed nm None (extracted a) t ;;
                 tmExistingInstance nm;;ret t
  | None => ret t
  end.
\end{lstlisting}

\subsection{Generation of Scott encodings} \label{sec:extract-constructor}

We use Scott encodings~\cite{Mogensen1992,Jansen2013} to encode inductive types and its constructors.
Scott encodings represent the matches on the inductive type.
For instance, the Scott encoding of the booleans are $\encf{\bool} {\texttt{true}} = \lambda x y.x$ and $\encf{\bool} {\texttt{false}} = \lambda x y.y$.
For natural numbers, the encodings are $\encf{\nat} 0 = \lambda z s.z$ and $\encf{\nat} (S n) = \lambda z s.s (\encf{\nat} n)$.

As before, we use type classes to remember previously generated encodings:
\begin{lstlisting}
Class encodable (A : Type) := enc_f : A -> L.term.  
Class registered (A : Type) := mk_registered
  { enc :> encodable A ;             (* the encoding function for A *)
    proc_enc : forall a, proc (enc a);  (* encodings are procedures *)
    inj_enc : injective enc          (* encoding is injective *) }.
\end{lstlisting}

For an inductive type with $n$ constructors, the constructor of
index $i$ which takes $a$ arguments has Scott encoding
$\texttt{gen\_constructor}\ a\ n\ i := \lambda x_1 \dots x_a.\lambda y_1 \dots y_n. y_i x_1 \dots x_a$. 

For natural numbers (a type with two constructors, i.e. $n = 2$), the constructor $S$ (which has index $i =
1$ and takes one argument, i.e. $a = 1$) has encoding $\lambda x. \lambda y_1 y_2. y_2 x$ (or $\lambda \lambda \lambda (0 2)$).

We use \coq{gen_constructor} to define a monadic operation \coq{tmExtractConstr}. 
If we want to extract \coq{map}, we first extract the two constants occurring in its definition (i.e.\ \coq{nil} and \coq{cons}) and then the actual function, always fully applied to their type parameters:
\begin{lstlisting}
Section Fix_X_Y.
  Context { X Y : Set }. Context { encY : encodable Y }.

  Run TemplateProgram (tmExtractConstr "nil_term" (@@nil X)).
  Run TemplateProgram (tmExtractConstr "cons_term" (@@cons X)).
  Run TemplateProgram (tmExtract "map_term" (@@map X Y)).
End Fix_X_Y.
\end{lstlisting}

\subsection{Generation of Encoding Functions}\label{sec:gen-enc}

We restrict our generation of encoding functions to simple inductive types of the form
\begin{lstlisting}
  Inductive T (X1 ... Xp : Type) : Type :=
  (* ... *) | constr_i_T : A1 -> ... -> An -> T X1 ... Xp | (* ... *).
\end{lstlisting}

where \coq{Aj} for $1 \leq \texttt{j} \leq \texttt{n}$ is either encodable or exactly \coq{T X1 ... Xn}.

For a fully instantiated inductive type \coq{B = T X1 ... Xp} with \coq{n} constructors we define the encoding function $\encf {\coqm B}$ as follows:
\begin{lstlisting}
fix f (b : B) := match b with 
 | ... | constr_i_T (x1 : A1) ... (xn : An) =>  $\lambda y_1 \dots y_{\texttt p}. y_{\texttt i}$ (f1 x1 ) ... (fn xn) | ... end
\end{lstlisting}
where \coq{fj} for $1 \leq \coqm{j} \leq \coqm{n}$ is a recursive call \coq{f} if \coq{Aj = B}, or $\encf {\coqm{Aj}}$ otherwise.
We implement a monadic function \coq{tmEncode} which can be used like this:
\begin{lstlisting}
Section Fix_X.
  Variable (X:Type). Context {intX : registered X}.
  Run TemplateProgram (tmEncode "list_enc" (list X)).
End Fix_X.
\end{lstlisting}

Note that in principle, more types are Scott-encodable, but we leave the
automatic generation for those types to future work.

\subsection{Extraction in Coq}

To be able to connect extracts $t_a$ to terms $a$ using the predicates $\computes a {t_a}$ and $\computes[\timec a] a {t_a}$ we define two type classes:
The class \coq{computable} is parameterised over $a$ and contains an extracted term $t_a : \term$ and a proof of $\computes a {t_a}$.
The class \coq{computableTime} is in addition parametrised over a time complexity function $\timec{a}$:
\begin{lstlisting}
Class computable {A : Type} {ty : $\TT$ A} (a : A) : Type :=
  { ext :> extracted a; 
    extCorrect : computes ty a ext }.

Class computableTime {A : Type} (ty : $\TT$ A) (a : A): Type :=
  { extT : extracted a; evalTime : $\mathcal{C}$ A ;
    extTCorrect : computesTime ty a extT evalTime }.
\end{lstlisting}

This way, we can write \coq{ext a} or \coq{extT a} for previously extracted terms $t_a$.
Note that since all relevant information can be obtained through the paramaters and the types of the fields, we can leave all instances of this classes opaque in Coq.

\section{Automated Verification}
\label{sec:verification}

We now give an overview over the set of tactics we provide in our framework.
All tactics are written in Ltac only, but some of them use the monadic operations explained in the last section.
We first explain the tactics to simplify \L-terms.
We then show how to register inductive datatypes to be used with the framework.
Lastly, we explain how to prove the computability relation $\computes a {t_a}$ and infer recurrence equations for a time bound $\timec a$.

\subsection{Symbolic Simplification for \L}\label{sec:symb-simp}

All tactics in this section are concerned with proving goals of the form ``$s$ is a procedure'' or ``$s$ reduces to $t$'', or transforming a goal like ``$s$ reduces to $t$'' to ``$s'$ reduces to $t$'' by simplifying $s$ to $s'$.
While all terms $s$ we simplify will be closed, they might not be concrete terms, e.g.\ contain the encoding of an arbitrary natural number.
The tactics will not unfold definitions.

\textbf{\coq{Lproc}:} The tactic \coq{Lproc} can prove that a term is closed, an abstraction or a procedure. It syntactically decomposes the term and uses a hint database for easier extensibility.%

\textbf{\coq{Lbeta}:} The tactic \coq{Lbeta} simplifies \L-terms by reducing all $\beta$-redices of the form $(\lambda s) t$ which are visible without unfolding definitions.
It uses \coq{Lproc} to show that $t$ is a procedure and that folded definitions used in $s$ are closed, thus left unchanged by the substitution.
\coq{Lbeta} is implemented by reflection, treating names as opaque and using closures to evaluate big terms more efficiently.
It can keep track of the number of beta-reductions performed.
For example, it simplifies the \L-term $(\lambda x y.\,x y y)\,u\,v$ in 2 steps to $u\,v\,v$.

\textbf{\coq{Lrewrite}:} The tactic \coq{Lrewrite} simplifies terms by the use of a hint database with the same name, containing the correctness statements for previously extracted terms, and by the use of local assumptions, which are important for recursion.
For efficiency reasons, it does not use Coq's built-in rewriting and instead traverses terms to find subterms where a hint from the database is applicable.
For example, it simplifies the \L-term $t_{+} (t_{+} (\encf \nat\ x) (\encf \nat\ 5)) (\encf \nat\ y)$ to $\encf \nat (x + 5 + y)$.
While traversing, \coq{Lrewrite} replaces occurences of $t_y$ with $y:Y$ of registered type by the trivial instance with extraction $\encf Y y$.
This guarantees canonicity of instances of \coq{computable} for registered types.

Additionally, \coq{Lrewrite} simplifies $t_f\,t_x$ to $t_{f x}$ for $x:X$ and $f:X\to Y$.
The concrete instance of $\coqm{computable} (f x)$ is constructed by combining the instances for $f$ and~$x$.

\textbf{\coq{Lsimpl}:} The tactic \coq{Lsimpl} repeatedly applies \coq{Lbeta} and \coq{Lrewrite} in alternation and can solve trivial goals by reflexivity.

\textbf{Time bounds:} All tactics can be used to analyse time bounds as well: \coq{Lbeta}, \coq{Lrewrite}, and \coq{Lsimpl} transform goals of the form $s \red^{\coqm{?k}} t$ to goals of the form $s' \red^{\coqm{?k}} t$ for an $s'$ with $s \red^{\coqm{k1}} s'$, instantiating the existential variable \coq{?k} with \coq{k1 + ?k'}.

\subsection{Registering Inductive Datatypes}\label{sec:add-ind-type}

To register an inductive datatype we provide the monadic operation \coq{tmGenEncode : ident -> Type -> TemplateMonad unit}:
\begin{lstlisting}
Run TemplateProgram (tmGenEncode "nat_enc" nat).
Hint Resolve nat_enc_correct : Lrewrite.
\end{lstlisting}
The operation generates the encoding function and three obligations, which are discharged automatically\rlap.\footnote{Using \coq{Global Obligation Tactic} of the \coq{Program} mode shipped with Coq.}
The first and second obligation regard procedureness and injectivity of the generated encoding function by tactics \coq{register_proc} and \coq{register_inj}.

The third obligation is saved as \coq{nat_enc_correct} and is generated similarly to the encoding function.
It states that the encoding behaves like Scott encoding and is also proven automatically, using the tactic \coq{extract match}.
In the case of natural numbers, it has the following type: \coq{nat_enc_correct : forall (n:nat) (s t:term), proc s -> proc t -> enc n s t $\red^{\leq 2}$ match n with 0 => s | S n' => t (enc n') end}.
The lemma has to be registered in the hint database \coq{Lrewrite} manually in order to be used by our tactics.

To work with an inductive type, a user also has to extract its constructors.
The constant constructors (e.g.\ $0$ for natural numbers) are trivially computable by their encoding:
\begin{lstlisting}
Instance reg_is_ext ty (R : registered ty) (x : ty) : computable x.
Proof.  exists (enc x). reflexivity. Defined.
\end{lstlisting}

A specific instance is only needed for the functional constructors of inductive data types:
\begin{lstlisting}
Instance term_S : computable S. Proof.  extract constructor. Qed.
\end{lstlisting}
The \coq{extract constructor} tactic extracts constructors as described in \cref{sec:extract-constructor} and show their correctness fully-automatically as described in the next section.%

\subsection{Automatically Proving Correctness} \label{sec:derive-correctness}
\setCoqFilename{Tactics.ComputableDemo}

As an example\footnote{Available as an interactive example in \coqlink[correctness_example]{\coq{Tactics/ComputableDemo.v} as \coq{Example correctness_example}}}, we take the boolean disjunction \coq{orb x y := if x then true else y}.
For the user, the extraction is fully automatic:
\begin{lstlisting}
Instance term_orb : computable orb. Proof.  extract. Qed.
\end{lstlisting}

The tactic \coq{extract} first extracts the Coq term as described in \cref{sec:extract-term}.
In this case, the result is $\lam{xy}{x (\coqm{ext true}) y}$.
The verification is then performed by iterating the tactic \coq{cstep}, where in each step a goal is of the form $\computes {\coqm{f}} s$.
The tactic \coq{cstep} performs simplifications depending on the Coq term \coq{f}.

Here, the initial proof goal reads as follows:
\begin{display}
  \computes {(\text{\coq{fun x y => if x then true else y}})} {(\lam{xy}{x (\coqm{ext true}) y})}
\end{display}

In case the Coq term is of function type and not syntactically a \coq{fix}, \coq{cstep} uses the definition of $\computes{}{}$ on function types and assumes a boolean \coq{x} computed by a term \coq{ext x}.
This yields as intermediate goal the existence of a procedure $v$ with
\begin{display}
  (\lam{xy}{x (\coqm{ext true}) y}) (\coqm{ext x}) \red^* v \text{ and } \computes {(\text{\coq{fun y => if x then true else y}})} {v}
\end{display}

Now \coq{cstep} uses \coq{Lsimpl} to derive $v$ by simplifying the term $(\lam{xy}{x (\coqm{ext true}) y}) (\coqm{ext x})$ to $\lam{y}{(\coqm{ext x}) (\coqm{ext true}) y}$, yielding the proof goal
\begin{display}
  \computes {(\text{\coqm{fun y => if x then true else y}})} {\lam{y}{(\coqm{ext x}) (\coqm{ext true}) y}}
\end{display}

The next call of \coq{cstep} assumes a fixed boolean \coq{y} and simplifies by \mbox{\coq{Lrewrite}:}
\begin{display}
  \computes {\text{\coqm{if x then true else y}}} {\text{\coqm{if x then ext true else ext y}}}
\end{display}

In case the Coq term syntactically has a case distinction on top, \coq{cstep} performs the same case distinction for the proof, here leaving the two goals $\computes {\text{\coqm{true}}} {\text{\coq{ext true}}}$ and $\computes {\text{\coq{y}}} {\text{\coq{ext y}}}$.
In both cases the Coq term is of registered type and the next call of \coq{cstep} proves these goals using the definition of $\computes{}{}$.

\subsubsection{Recursive Functions}\label{sec:derive-rec-fun}

Recall that recursive functions in Coq are defined via the \coq{fix} (or \coq{Fixpoint}) construct, which allows the application of recursive calls to `smaller' arguments, where the notion `smaller' is due to the guardedness checker of Coq.
The tactic \coq{cstep} proves the correctness using \coq{fix} as well, with the same recursive calls as the extracted function.
Therefore, the guardedness checker will accept the proof for exactly the same reasons it accepted the function definition\footnote{The guardedness checker rejects some of our produced proofs when extracting functions not directly structurally recursive: This is due to the additional heuristics in the guardedness checker. }.

As an example\footnote{Available as an interactive example in \coqlink[correct_recursive]{\coq{Tactics/ComputableDemo.v} as \coq{Example correct_recursive}}}, the extraction of \coq{map A B} (see \cref{fig:map_vec}) for registered types \coq{A} and \coq{B} is of shape
$\lam f {\rho~v_1}$ for a procedure $v_1$, where $\rho$ is the fixed-point combinator from \Cref{lem:rho}.

To verify this term, the proof goal is
\begin{display}
  \computes {\coqm{fun f => fix map l := (...)}} {\lam f {\rho~v_1}}
\end{display}

The first call of \coq{cstep} is as in \cref{sec:derive-correctness} and yields the following goal, where $v_2$ is obtained by replacing the \L-variable $f$ with \coq{ext f} for a fixed computable \coq{f : A -> B} in~$v_1$:
\begin{display}
  \computes {\coqm{fix map l := (...)}} {(\rho~v_2)}
\end{display}

In case the Coq term is syntactically a \coq{fix}, \coq{cstep} uses the definition of $\computes {} {}$ on function types, but generalises the goal over all arguments of \coq{fix} (in this case only \coq{l}):
\begin{display}
  \forall \coqm{l}. \Sigma v : \term . {(\rho~v_2) (\coqm{ext l})} \red^* v \land \computes {\coqm{(fix map l := (...))} \coqm{(ext l)}} v
\end{display}

\coq{cstep} now inserts a a \coq{fix} into the proof term, obtaining an inductive hypothesis \coq{IH} of the same type as the goal.
For the proof term to type-check in the end, \coq{IH} can only be used on arguments structurally smaller than \coq{l}.
To guarantee this, \coq{cstep} always performs a case analysis on the recursive argument first, i.e.\ in this case on \coq{l}, yielding two goals.

In both resulting cases, \coq{cstep} calls \coq{Lrewrite} which uses the inductive hypothesis~\coq{IH} to simplify all occurrences of $(\rho~v_2) (\coqm{ext l'})$ to \coq{ext ((fix map l := (...)) l')}. %
In both goals, \coq{cstep} needs to obtain a procedure $v$ with ${(\rho~v_2) (\coqm{ext l})}$, which is done using \coq{Lsimpl}.
For $\coqm{l}=[]$, the goal is trivial because $\computes {[]} {\coqm{ext []}}$.
\enlargethispage{2\baselineskip}
In the recursive case \coq{l = x :: l'}, \coq{Lsimpl} yields the trivial goal
\begin{display}
  \computes {\coqm{f x :: ((fix map l := (...)) l')}} {\coqm{ext (f x :: ((fix map l := (...))  l')}}
\end{display}%

\subsubsection{Higher-Order Functions}\label{sec:derive-higher-order}

Terms containing higher-order functions applied to arguments need a syntactic transformation to be supported by our framework.
To verify the correctness of e.g.\ \coq{map (fun x => x + y) l} as part of a bigger program, we essentially need to show
\begin{display}
  t_\coqm{map}\,(\lambda x.t_{\coqm{+}}\,x\,y)\,(\encf {} {\coqm{l}}) \sim \coqm{map (fun x => x + y) l}
\end{display}

To use the definition of $\computes {} {}$ for $t_{\coqm{map}}$, we would have to show $(\lambda x.t_{\coqm{+}}\,x\,y) \sim \coqm{(fun x => x + y)}$.
This introduces several difficulties, one is that the term might contain free variables that need to be beta abstracted, and another one occurs when time bound are of interest:
Since our verification of time bounds is only semi-automatic and requires the user to instantiate the recurrences by hand, we would need to interrupt the proof here for a user to fill in the concrete time bounds for $(\lambda x.t_{\coqm{+}}\,x\,y)$.

We thus restrict the scope of the framework and only cover applications of higher order functions to arguments which syntactically are composed from previously extracted term by application (without the use of abstractions).
In this case this would mean that one has to define a Coq term \coq{f y := fun x => x + y}, which has to be extracted before \coq{map (f y) l}.

\subsection{Proving Time Bounds}\label{sec:derive-time-bounds}

All simplification tactics also keep track of the number of $\beta$-steps in reductions and can thus be used to infer recurrence equations a correct time complexity function has to satisfy.
The only obligation left to the user when proving instances of \coq{computableTime} is to provide a solution to this recurrence equations.
As an example, we consider boolean disjunction again and want to find a time complexity function $\tau : \bool \to \unit \to \nat \times (\bool \to \unit \to \nat \times \unit)$:
\begin{lstlisting}
Instance term_orb : computableTime orb $\tau$.
Proof. extract.  
\end{lstlisting}
This leaves the user with the recurrence equations $\pi_1 (\tau x \star) \geq 1$ and $\pi_1(\pi_2(\tau x \star) y \star) \geq 3$, indicating that $t_{\coqm{orb}}$ needs one step to reduce to an abstraction if applied to an encoded boolean $x$ and this abstraction needs 3 further steps to a value if applied to a boolean $y$.
Thus, choosing $\tau$ as \coq{fun _ _ => (1,fun _ _ => (3,tt))} works.
We provide the tactic \coq{solverec} which simplifies goals containing inequations and tries to show them using the \coq{lia} tactic shipped with Coq. If proving the inequality needs further reasoning, the tactic presents the user with simplified goals.

The recurrence equations for the time bound are inferred incrementally by \coq{cstep} using an existential variable. 
To prove \coq{computableTime orb $\tau$}, \coq{cstep} first introduces an assumption \coq{H : ?P $\tau$} and opens a new goal \coq{?P $\tau$}.
In each step, \coq{cstep} performs the transformations described in \cref{sec:derive-correctness} while keeping track of the number of steps, asserting that \coq{$\tau$} needs to be larger than the number of $\beta$-steps performed by instantiating \coq{?P} further.
For non-recursive functions, this will only produce lower bounds for components of $\tau$, while for recursive correctness proof it produces inequalities that contain $\tau$ on both sides.

To find time bound functions interactively, we define the opaque polymorphic constant \coq!cnst {X:Type} (x:X) : nat := 0! which can be used as a place-holder for unknown constants.
To find the time complexity for \coq{map}\footnote{Available as an interactive example in \coqlink[comeUp_timebound]{\coq{Tactics/ComputableDemo.v} as \coq{comeUp_timebound}}} one would start with the following:
\begin{lstlisting}
Lemma termT_map A B (Rx : registered A) (Ry: registered B):
    computableTime (@@map A B) (fun $f$ $\tau_f$ => (cnst "c",fun $l$ _ => (cnst ("g",$l$),tt))).
Proof. extract. solverec. 
\end{lstlisting}
This yields three conditions: \coq{1 <= cnst "c"}, \coq{7 <= cnst ("g", [])}, and \coq{fst ($\tau_f$ a tt) + cnst ("g", l) + 11 <= cnst ("g", a :: l)}.
Note that \coq{cnst} allows us to keep track of the different arguments that the time bound is instantiated with later.
As expected, the time bound of map must also sum up all the time bounds for calling \coq{f} on all elements of the list, and indeed, \coq{solverec} can show the lemma using this time bound:
\begin{lstlisting}
fun $f$ $\tau_f$ => (1, fun $l$ _ => (fold_right (fun $x$ res => $\pi_1$ ($\tau_f$ $x$ tt) + res + 11) 7 $l$, tt))
\end{lstlisting}

\section{Case studies}
\label{sec:casestudy}
We provide three case studies:
A universal \L-term obtained as extraction from the Coq definition of a step-indexed self-interpreter for \L (in \href{https://uds-psl.github.io/certifying-extraction-with-time-bounds/website/L.Functions.Universal.html}{\texttt{Functions/Universal.v}}), a many-one reduction from solvability of Diophantine equations to the halting problem of \L (in \href{https://uds-psl.github.io/certifying-extraction-with-time-bounds/website/L.Reductions.H10.html}{\texttt{Reductions/H10.v}}), and a linear simulation of Turing machines in \L (in \href{https://uds-psl.github.io/certifying-extraction-with-time-bounds/website/L.TM.TMEncoding.html}{\texttt{TM/TMEncoding.v}}).

\subsection{Step-indexed \L-interpreter}
\label{sec:self}
\newcommand\evalf{\texttt{eva}}

A step-indexed interpreter for \L is a function $\evalf : \nat \to \term \to \Option \term$ s.t. for closed $s$ we have $(\exists n.~\evalf~n~s = \some t) \toot (s \red^* t \land t \textit{ is a procedure})$.
The function can be defined as follows~\cite{FS}:
\begin{lstlisting}
Fixpoint eva (n : nat) (u : term) := 
  match u with
  | var n => None | lam s => Some (lam s)
  | app s t => match n with 
              | 0   => None
              | S n => match eva n s, eva n t with
                      | Some (lam s), Some t => eva n (subst s 0 t)
                      |   _ , _ => None
  end           end     end.
\end{lstlisting}
Here \coq{subst s 0 t} denotes substitution, which uses \coq{Nat.eqb} as boolean equality test on natural numbers.
We extract all three functions in reverse order.
To do so, we first need encodings for natural numbers and term constructors as shown in \Cref{sec:add-ind-type} and encodings for terms.
We first generate the encoding function and register it:
\begin{lstlisting}
Run TemplateProgram (tmGenEncode "term_enc" term).
Hint Resolve term_enc_correct : Lrewrite.
\end{lstlisting}
We can then extract the non-constant constructors, \coq{Nat.eqb}, \coq{subst}, and \coq{eva}:
\begin{lstlisting}
Instance term_var : computableTime var (fun n _ => (1, tt)).
Proof. extract constructor. solverec. Qed.
Instance term_app : computableTime app (fun s1 _ => (1, (fun s2 _ => (1, tt)))).
Proof. extract constructor. solverec. Qed.
Instance term_lam : computableTime lam (fun s _ => (1, tt)).
Proof. extract constructor. solverec. Qed.

Instance termT_nat_eqb :
  computableTime Nat.eqb (fun x _ => (5, (fun y _ => ((min x y) * 15 + 8, tt)))).
Proof.  extract. solverec. Qed.

Instance term_substT :
  computableTime subst (fun s _ => (5, (fun n _ => (1, (fun t _ => 
                           (15 * n * size s + 43 * (size s) ^ 2 + 13, tt)))))).
Proof.  extract. solverec. Qed.

Instance term_eva : computable eva.
Proof. extract. Qed.
\end{lstlisting}

Note that the implementation of $\texttt{eva}$ is very naive and needs steps exponential in $n$, we thus omit its time complexity\rlap.\footnote{The recurrence equation generated for $\text{eva}$ one would have to solve reads $f (1+n) (s_1 s_2) \geq f~n~ s_1 + f~n~s_2 + 43 \cdot {(\texttt{size}~t_1)}^ 2 + f ~n~ (\subst {t_1} 0 {t_2}) + 53$, with $\texttt{eva}~n~s_1 = \lambda t_1$ and $\texttt{eva}~n~s_2 = t_2$.}
A more reasonable implementation could be obtained by extracting the heap-based abstract machine from~\cite{kunze2018formal} to \L.

\subsection{Diophantine equations}
\newcommand\LHalt{\ensuremath{\mathcal{E}}}

The problems contained in the library of undecidable problems in Coq~\cite{forster2018towards} are proven undecidable by a chain of many-one reductions starting at the halting problem for Turing machines.
As a matter of fact, all problems contained in the library so far are actually interreducible.
An easy way to prove this is to reduce leafs in the reduction graph to the halting problem for L defined as $ \LHalt s := \exists v.(s \red^* v \land v \textit{ is an abstraction})$ and then implement one general reduction from \LHalt~to the halting problem of Turing machines.

As an example how to reduce problems to \LHalt~we use our framework to reduce solvable Diophantine equations~\cite{Larchey-WendlingForster:2019:H10_in_Coq}, i.e. Hilbert's tenth problem \textsf{H10}, to \LHalt.

We first explain the general structure using mathematical notation.
In~\cite{forster2019synthetic}, the authors define synthetic notions of decidability and enumerability.
If this definitions are enriched with explicit computability assumptions, one obtains:
\begin{definition}
  A predicate $p : X \to \Prop$ is \L-decidable if there exists \emph{a computable} $f : X \to \bool$ s.t. $\forall x.~p x \toot f x = \mathsf{tt}$.
\end{definition}
\begin{definition}
  A predicate $p : X \to \Prop$ is \L-enumerable if there exists \emph{a computable} $f : \nat \to \Option X$ s.t. $\forall x.~p x \toot \exists n.f n = \some x$.
\end{definition}
\setCoqFilename{Computability.Synthetic}%
\begin{theorem}[][L_enumerable_halt]
  If $p : X \to \Prop$ is \L-enumerable and equality on $X$ is \L-decidable, then $p \preceq \mathcal{E}$.
\end{theorem}
\begin{proof}
  Let $f$ be the (computable) enumerator $\nat \to \Option X$ and $d : X \times X \to \bool$ the (computable) equality decider.
  We define $s := \lambda x.~\mu (\lambda n.~t_f n~(\lambda y.t_d~x~y)~t_{\mathsf{ff}})$.
  Here, $\mu$ is an unbounded search operator, i.e.\ $s$ performs unbouded search for $x$ in the range of $f$.
  Then $p x \toot \mathcal{E}(s ~\overline x)$.
\end{proof}
Moreover, it is easier to implement concrete enumerators based on lists, i.e.\ computable enumerators $f : \nat \to \List X$ s.t. $p x \toot \exists n.~x \in f n$.
The equivalence proof of both notions can be found in~\cite{forster2019synthetic}.
Extending the proof with explicit computability assumptions as needed here is straightforward and we refer to the Coq code.

We now switch to a more technical notation and show how to construct such a list enumerator for \textsf{H10} in Coq.
We first define the type of polynomials, generate its encoding and extract its constructors:
\begin{lstlisting}
Inductive poly : Set :=
    poly_cst : nat -> poly             | poly_var : nat -> poly 
  | poly_add : poly -> poly -> poly | poly_mul : poly -> poly -> poly.

Run TemplateProgram (tmGenEncode "enc_poly" poly).
Hint Resolve enc_poly_correct : Lrewrite.

Instance term_poly_cst : computable poly_cst. extract constructor. Qed.
Instance term_poly_var : computable poly_var.   extract constructor. Qed.
Instance term_poly_add : computable poly_add.   extract constructor. Qed.
Instance term_poly_mul : computable poly_mul.   extract constructor. Qed.
\end{lstlisting}
We define evaluation of polynomials under assignments \coq{S : list nat} as and the decision problem \coq{H10} as follows:
\begin{lstlisting}
Fixpoint eval (p : poly) (S : list nat) :=
  match p with
  | poly_cst n => n
  | poly_var n => nth n S 0
  | poly_add p1 p2 => eval p1 S + eval p2 S
  | poly_mul p1 p2 => eval p1 S * eval p2 S
  end.
Definition H10 '(p1, p2) := exists S, eval p1 S = eval p2 S.
Instance term_eval : computable eval. extract. Qed.
\end{lstlisting}
where \coq{nth n S d} returns the \coq{n}-th element in \coq{S}, or \coq{d} if \coq{S} is not long enough.
We also define a computable function \coq{poly_eqb : poly -> poly -> bool} deciding syntactic equality.

To show that \textsf{H10} is \L-enumerable, we enumerate all polynomials using \coq{L_poly : nat -> list poly}.
Due to the restriction that higher-order arguments can not syntactically contain abstractions, we first extract uncurried versions of the constructors:
\begin{lstlisting}
Definition poly_add' '(x,y) : poly  := poly_add x y.
Instance term_poly_add' : computable poly_add'. extract. Qed.

Definition poly_mul' '(x,y) : poly := poly_mul x y.
Instance term_poly_mul' : computable poly_mul'. extract. Qed.

Fixpoint L_poly n : list (poly) :=
  match n with
  | 0 => []
  | S n => L_poly n  ++  map poly_cst (L_nat n) ++ map poly_var (L_nat n)
                      ++ map poly_add' (list_prod (L_poly n) (L_poly n))
                      ++ map poly_mul' (list_prod (L_poly n) (L_poly n))
  end.
  
Instance term_L_poly : computable L_poly. extract. Qed.
\end{lstlisting}

The last and crucial lemma is the adaption of Fact 2.9 from~\cite{forster2019synthetic}:

\begin{lemma}[][projection]
  If $p : X \times Y \to \Prop$ is \L-enumerable, then $\lambda x.\exists y.~p(x,y)$ is \L-enumerable.
\end{lemma}
\setCoqFilename{Reductions.H10}%
\begin{theorem}[][H10_enumerable]
  \textsf{\emph{H10}} is \L-enumerable.
\end{theorem}
\begin{proof}
  By \Cref{coq:projection} we have to give a list enumerator for two polynomials \coq{p1} and \coq{p2} together with solutions \coq{S}:
\begin{lstlisting}
fix f n := match n with 0 => [] 
  | S n => f n ++ filter (fun '(p1,p2,S) => Nat.eqb (eval p1 S) (eval p2 S)) 
                    (list_prod (list_prod (L_poly n) (L_poly n)) (L_list_nat n)) end.
\end{lstlisting}
  where \coq{list_prod} is the cartesian product on lists and \coq{L_list_nat} is a list enumerator for \coq{list nat}.
\end{proof}

\begin{corollary}[][H10_converges]
  $\textsf{\emph{H10}} \preceq \mathcal{E}$
\end{corollary}
\begin{proof}
  By Theorems~\ref{coq:H10_enumerable} and~\ref{coq:L_enumerable_halt}.
\end{proof}

\subsection{Turing Machines}

We show how our framework can be used to reduce the halting problem of multi-tape Turing machines $\texttt{Halt}$ to the halting problem of \L.
We employ a Coq implementation of the definition of Turing machines by Asperti and Ricciotti~\cite{asperti2015}, who formalise Turing machines in Matita.
\begin{lstlisting}
Definition loopM : forall (sig : finType) (n : nat) (M : mTM sig n),
   mconfig sig (states M) n -> nat -> option (mconfig sig (states M) n) := (* ... *)

Definition Halt :{ '(Sigma, n) : _ & mTM Sigma n & tapes Sigma n} -> _ :=
  fun '(existT2 _ _ (Sigma, n) M tp) =>
    exists (f: mconfig _ (states M) _), halt (cstate f) = true
                                   /\ exists k, loopM (mk_mconfig (start M) tp) k = Some f.
\end{lstlisting}

Their formalisation uses the (dependent) vector type to model multiple tapes and an explicit transition function.
Both aspects do not fit in our framework directly.
We thus showcase two techniques to extend our framework in certain cases.

First, to encode types not in the scope of the framework, we notice that an encoding for a type $A$ can be obtained from an encoding function $\encf B$ given an injective function $A \to B$.
We pack this insight in the definition \coq{registerAs}, which can be used as follows:
\begin{lstlisting}
Instance register_vector X `{registered X} n : registered (Vector.t X n).
Proof. apply (registerAs VectorDef.to_list). (* injectivity proof *) Defined.
\end{lstlisting}

Second, we observe that computability is closed under extensional equality:

\setCoqFilename{Tactics.Computable}
\begin{definition}[][extEq]
  We define extensional equality for a type $A$ with $\texttt{ty} : \TT{A}$ recursively on $\texttt{ty}$.
  Elements $x, y$ of an encodable type $A$ are extensionally equal if they are equal.
  Functions $f, g : A \to B$ are extensionally equal if for all $a : A$, $f a$ is extensionally equal to $g a$.
\end{definition}

\begin{lemma}[][computesExt]
If $f$ and $g$ are extensionally equal and $\computes[\timec {}] f {t}$ then $\computes[\timec {}] g t$.
\end{lemma}

Combining those two insights allows us to extract any vector operation by extracting the corresponding list-operation.

Furthermore, we use the fact that functions with finite domain and co-domain can always be translated into a value table containing lists of pairs.
We can thus show that every transition function is computable in time independent of the current configuration, and derive time bound for \coq{loopM}, executing a machine for $k$ steps:
\begin{lstlisting}
Instance term_trans : computableTime (trans (m:=M)) (fun _ _ => (transTime,tt)).
Proof. (* ... *) Qed.

Instance term_loopM :
  let c1 := (haltTime + n*121 + transTime + 76) in let c2 := 13 + haltTime in
  computableTime (loopM (M:=M)) (fun _ _ => (5,fun k _ => (c1 * k + c2,tt))).
Proof. unfold loopM. extract. solverec. Qed.
\end{lstlisting}

Here \coq{haltTime} and \coq{transTime} are constants depending on the concrete machine, its number of tapes and its alphabet.
By unbounded search over all number of steps $k$ we obtain:

\setCoqFilename{Reductions.TM}
\begin{theorem}[][Halt_eva]
  \emph{${\coqm{Halt}}$} reduces to $\mathcal{E}$.
\end{theorem}

\section{Conclusion}\label{sec:future}

\subparagraph*{Formalisation}

The tools in our framework heavily rely on Coq's tactic language Ltac to verify the correctness of extracted terms.
During the verification, existential variables are crucial to generate the recurrence equations described in \cref{sec:derive-time-bounds} while simultaneously simplifying the L-terms as described in \cref{sec:symb-simp}.
For this simplification, we implement a reflective simplification tactic for L-terms used in \coq{Lbeta}.
We tried to use setoid-rewriting for \coq{Lrewrite}, but the need to track the number of reduction steps requires us to implement our own, domain-specific rewriting tactic in Ltac.
This tactic implements bottom-up rewriting, resulting in smaller proof terms and faster rewriting, by performing many rewrite steps in one pass through the term: A tactic using congruence lemmas descends in the term and on the way out, rewriting steps are performed. We use the hint databases for the \coq{auto}-tactic to add new lemmas for rewriting.

Typeclasses are employed as a kind of dictionary, e.g. to look up the extraction for a previously extracted function or its correctness lemma.

The framework consists of roughly 2100 lines of code, of which 370 are for the definitions described in \cref{sec:time-bounds} and their properties, 380 are for the extraction in \cref{sec:extract}, 950 are for the simplification presented in \cref{sec:symb-simp}, and 420 are for the tactics proving those extracts correct in \cref{sec:derive-correctness}.

In total, the case studies consist of 340 lines of specification and 280 lines f code: 20 lines are for the universal machine, 200 for H10 and 400 for the Turing machine interpreter.
All examples are built on a library of extracted functions concerning natural numbers, booleans and lists, which consists of 360 lines of code.

\subparagraph*{Future Work}

There are several directions in which the framework can be extended.
We would like to extend the framework to support space bounds in addition to time bounds, based on the space measure defined in~\cite{FKR}.
Furthermore, our automation framework is sound by construction, because it produces proofs.
We conjecture it to be complete for the described fragment of Coq's type theory we are considering, but reasoning about tactics programmed in Ltac is basically impossible.
In the future, we would like to be able to support all of Coq's type theory (possibly leaving out co-inductive types).
In order to do that, the extraction process would have to support proof and type erasure, which can be implemented using Template-Coq.

On the more conceptual side, our extraction basically returns realisers in a realisability model for the treated fragment of Coq's type theory.
We would like to analyse and verify such realisability models using MetaCoq, possibly connecting the (weak call-by-value) evaluation relation defined in MetaCoq with reduction in the realisability model, yielding a proof that for a certain subset of Coq's type theory, all definable functions are indeed computable.

Lastly, we hope that our framework enables the formalisation of basic computational complexity theory in Coq.
We would like to mechanise results like a time hierarchy theorem for the call-by-value $\lambda$-calculus.
The commonly known proofs for Turing machines or similar models use self-interpreters.
The tightness of the provable gap then depends on the time-efficiency of the interpreter in use.
As mentioned, the self-interpreter given in \Cref{sec:self} is too inefficient and we want to extract the interpreters described in~\cite{kunze2018formal} and~\cite{FKR} to \L.

\bibliography{bib.bib}

\end{document}